%% file: positive.tex
\documentclass{LMCS}

\def\dOi{11(1:13)2015}
\lmcsheading%
{\dOi}
{1--21}
{}
{}
{Feb.~14, 2014}
{Mar.~27, 2015}
{}

\ACMCCS{[{\bf Theory of computation}]: Logic---Type theory; Semantics
  and reasoning---Program constructs---Type structure; Semantics and
  reasoning---Program semantics---Categorical semantics}


\usepackage{microtype}
\usepackage{amssymb,amsmath,amsthm}
\usepackage{stmaryrd}

\usepackage{mathtools} 
\usepackage{xspace}

\usepackage[all]{xy}

\usepackage{proof}

\usepackage{hyperref}
\usepackage[capitalise]{cleveref}

\include{macros}

\title{Positive Inductive-Recursive Definitions}
\author[N.~Ghani]{Neil Ghani}
\address{University of Strathclyde, UK}
\email{ng@cis.strath.ac.uc, \{lorenzo.malatesta,fredrik.nordvall-forsberg\}@strath.ac.uk}
\author[L.~Malatesta]{Lorenzo Malatesta}
\author[F.~N.~Forsberg]{Fredrik Nordvall Forsberg}
\thanks{This work was supported by the Engineering and Physical Sciences Research Council [grant numbers EP/G033056/1, EP/K023837/1]}

\begin{document}

\begin{abstract}  
  A new theory of data types which allows for the definition of types
  as initial algebras of certain functors $\Fam(\C) \rightarrow
  \Fam(\C)$ is presented. This theory, which we call \emph{positive
    inductive-recursive definitions}, is a generalisation of Dybjer
  and Setzer's theory of inductive-recursive definitions within which
  $\C$ had to be discrete --- our work can therefore be seen as
  lifting this restriction.  This is a substantial endeavour as we
  need to not only introduce a type of codes for such data types (as
  in Dybjer and Setzer's work), but also a type of morphisms between
  such codes (which was not needed in Dybjer and Setzer's
  development).  We show how these codes are interpreted as functors
  on $\Fam(\C)$ and how these morphisms of codes are interpreted as
  natural transformations between such functors. We then give an
  application of positive inductive-recursive definitions to the
  theory of nested data types and we give concrete examples of
  recursive functions defined on universes by using their elimination
  principle.  Finally we justify the existence of positive
  inductive-recursive definitions by adapting Dybjer and Setzer's
  set-theoretic model to our setting.
\end{abstract}

\maketitle

\section{Introduction} 
Inductive types are the bricks of dependently typed programming
languages: they represent the building blocks on which any other type
is built.  The mortar the dependently typed programmer has at her
disposal for computation with dependent types is recursion.  Usually,
a type $A$ is defined inductively, and then terms or types can be
defined recursively over the structure of $A$.  The theory of
inductive-recursive definitions~\cite{Dyb,DybSetz1} explores the
simultaneous combination of these two basic ingredients, pushing the
limits of the theoretical foundations of data types.

The key example of an inductive-recursive definition is Martin-L\"of's
universe \`a la Tarski~\cite{ML1}.  A type $U$ consisting of codes for
small types is introduced, together with a decoding function $T$,
which maps codes to the types they denote.  The definition is both
inductive and recursive; the type $U$ is defined inductively, and the
decoding function $T$ is defined recursively on the way the elements
of $U$ are generated.  The definition needs to be simultaneous, since
the introduction rules for $U$ refer to $T$.  We illustrate this by
means of a concrete example: say we want to define a data type
representing a universe containing a name for the natural numbers,
closed under $\Sigma$-types. Such a universe will be the smallest
family of sets $(U, T)$ satisfying the following
equations \begin{equation} \label{eq:U-closed-sigma} \begin{array}{lll}
    U & = & 1\;  + \; \Sigma\, u \! :  \! U. \;Tu \to U \\
    T (\inl\; *) & = & \nattype \\
    T (\inr\; (u,f)) & = & \Sigma x \!: \! T u. \;T (f x)\\
\end{array}
\end{equation}
In this definition we see how ground types and the type constructor
$\Sigma$ are reflected in $U$.  The left summand of the right hand
side of the equation defining $U$ is a code for the type of natural
numbers, while the right summand is a code reflecting $\Sigma$-types.
Indeed the name of a the type $\Sigma\,A\,B$ for $A:\Set$, $B:A\to
\Set$ in the universe $(U,T)$ will consists of a name in $U$ for the
type $A$, i.e.\ an element $u\! :\! U$, and a function $f:Tu \to U$
representing the $A$-indexed family of sets $B$.  The decoding
function $T$ maps elements of $U$ to types according to the
description above: the code for natural numbers decodes to the set of
natural numbers $\nattype$ while an element $(u,f)$ of the right
summand decodes to the $\Sigma$-type it denotes. Other examples of
inductive-recursive definitions have also appeared in the literature,
such as e.g.\ Martin-L\"of's computability predicates~\cite{ML0} or
Aczel's Frege structures~\cite{Aczel}.  Lately the use of
inductive-recursive definitions to encode invariants in ordinary data
structures has also been considered~\cite{Stevan}.

Dybjer's~\cite{Dyb} insight was that these examples are instances of a
general notion, which Dybjer and Setzer~\cite{DybSetz1} later found a
finite axiomatisation of. Their theory of inductive-recursive
definitions $\IR$ consists of: (i) a representation of types as
initial algebras of functors; and (ii) a grammar for defining such
functors. Elements of the grammar are called \IR codes, while functors
associated to \IR codes are called \IR functors.  The theory naturally
covers simpler inductive types such as lists, trees, vectors,
red-black trees etc.\ as well.  Dybjer and Setzer \cite{DybSetz2} then
gave an initial algebra semantics for \IR codes by showing that \IR
functors are naturally defined on the category $\Fam(\D)$ of families
of elements of a (possibly large) type $\D$ and that these functors do
indeed have initial algebras.  More generally, abstracting on the
families construction and the underlying families fibration
$\pi:\Fam(\D) \to \Set$, we have recently shown how to interpret $\IR$
functors in an arbitrary fibration endowed with the appropriate
structure~\cite{GHMNS}.  In this article, we will only consider the
families fibration.

There is, however, a complication.  When interpreting \IR\ functors
such as those building universes closed under dependent products, the
mixture of covariance and contravariance intrinsic in the $\Pi$
operator forces one to confine attention to functors $\Fam\, |\C|
\rightarrow \Fam\, |\C|$ or, equivalently, to work with only those
morphisms between families which are commuting triangles. More
abstractly, as we have shown~\cite{GHMNS}, this corresponds to working
in the split cartesian fragment of the families fibration
$\pi:\Fam(\C) \to \Set$, i.e.\ to only consider those morphisms in
$\Fam(\C)$ which represent strict reindexing. In this paper we remove
this constraint and hence explore a further generalization of $\IR$,
orthogonal to the one proposed in Ghani et al.~\cite{GHMNS}.  We
investigate the necessary changes of $\IR$ needed to provide a class
of codes which can be interpreted as functors $\Fam(\C) \rightarrow
\Fam(\C)$. This leads us to consider a new variation $\pIR$ of
inductive-recursive definitions which we call \emph{positive
  inductive-recursive definitions}. The most substantial aspect of
this new theory is that in order to define these new codes, one needs
also to define the morphisms between codes. This is no handle-turning
exercise!

We first recall Dybjer and Setzer's theory of inductive-recursive
definitions (\cref{sec:IR}).  To develop the theory we then introduce
a syntax and semantics consisting of $\pIR$ codes and their morphisms,
and an explanation how these codes are interpreted as functors
$\Fam(\C) \to\Fam(\C)$, where $\C$ is an arbitrary category
(\cref{sec:pos-ind-rec}).  We then illustrate the stronger elimination
principles that are possible for positive inductive-recursive
definitions. We consider several examples of catamorphisms that are
not possible with ordinary inductive-recursive definitions
(\cref{sec:elim-principle}).  As a practical application, we use
positive inductive-recursive definitions to shed new light on nested
data types (\cref{sec:nested-types}).  We formally compare $\pIR$ with
the existing theory of $\IR$ (\cref{sec:oldIR}), and adapt Dybjer and
Setzer's model construction to our setting
(\cref{sec:initial-alg}). The material in this paper has been
formalised in Agda~\cite{agda}.

The paper uses a mixture of categorical and type theoretic
constructions. However, the reader should bear in mind that the
foundations of this paper are type theoretic. In other words, all
constructions should be understood to take place in extensional
Martin-L\"of type theory with one universe $\Set$. This is entirely
standard in the literature. The one exception is the use of a Mahlo
cardinal required to prove that positive inductive recursive functors
have initial algebras in \cref{sec:initial-alg}. It should be
emphasised that the Mahlo cardinal is only used to justify the
soundness of the theory, and does not play any computational role.  We
refer the interested reader to Dybjer and Setzer~\cite{DybSetz1} ---
they use a Mahlo cardinal for the same purpose --- for the technical
details.  We also use fibrational terminology occasionally when we
feel it adds insight, however readers not familiar with fibrations can
simply ignore such comments.

\section{Inductive-recursive definitions}
\label{sec:IR}

In increasing complexity and sophistication, inductive definitions,
indexed inductive definitions and inductive-recursive definitions
encode more and more information 
about the data structures in question into the type itself. Being
situated at the top of this hierarchy, inductive-recursive definitions
provide a unifying theoretical framework for many different forms of
data types. Indeed, both inductive and indexed inductive definitions
are simple instances of $\IR$~\cite{Malatesta}.

The original presentation of induction-recursion given by
Dybjer~\cite{Dyb} was as a schema.  Dybjer and Setzer~\cite{DybSetz1}
further developed the theory to internalize the concept of an
inductive-recursive definition.  They developed a finite
axiomatization of the theory through the introduction of a special
type of codes for inductive-recursive definitions.  The following
axiomatization, which closely follows Dybjer and
Setzer~\cite{DybSetz1}, presents the syntax of $\IR$ as an inductive
definition.

\begin{definition}[$\IR$ codes]
\label{def:IR-codes}
Let $\D$ be a (possibly large) type. The type of $\IR(\D)$ codes has
the following constructors:
\[
\infer{\iota\,d\,:\,\IR(D)}{d\,:\D}
\]
\[
\infer{\sigma_A\,f\,:\IR(\D)}{A\,:\Set & f\,:A\rightarrow \IR(\D)}
\]
\[
\infer{\delta_A\,F\,:\IR(\D)}{A\,:\Set& F\,:(A\rightarrow \D)\rightarrow \IR(\D)}
\]
\end{definition}

This is the syntax of induction-recursion --- it is quite remarkable
in our opinion that this most powerful of theories of data types can
be presented in such a simple fashion. These rules have been written
in natural-deduction style and we may use the ambient type theory to
define, for example, the function $f$ in the code $\sigma_A f$.  An
example of an $\IR$ code is given in \Cref{ex:sigma}; this code
represents the universe containing the natural numbers and closed
under $\Sigma$-types given in Equation~\eqref{eq:U-closed-sigma}. We
now turn to the semantics of induction-recursion: we interpret $\IR$
codes as functors, and to this end, we use the standard families
construction $\Fam$ from category theory.  We start recalling the
definition of the category $\Fam(\C)$ of families of objects of a
category $\C$.

\begin{definition}
  Given a category $\C$, the category $\Fam(\C)$ has objects pairs
  $(X, P)$ where $X$ is a set and $P: X \to \C$ is a functor which we
  can think of as an $X$-indexed family of objects of $\C$.  A
  morphism from $(X,P)$ to $(Y,Q)$ is a pair $(h,k)$ where $h:X\to Y$
  is a function, and $k:P\stackrel{\cdot}{\rightarrow}Q\circ h$ is a
  natural transformation.
\end{definition}
Of course, the naturality condition in the definition of a morphism of
families is vacuous as the domains of the functors in question are
discrete.
\begin{remarks}
\label{rem:Fam}
For any category $\C$, the category $\Fam(\C)$ always has rich structure:
\begin{itemize}
\item $\Fam(\C)$ is fibred over $\Set$ (see
  e.g. Jacobs~\cite{Jacobs}).  We omit the definitions here, but
  recall the standard splitting cleavage of the fibration $\pi :
  \Fam(\C) \to \Set$ which is relevant later: a morphism $(h, k) : (X,
  P) \to (Y, Q)$ is a split cartesian morphism if $k$ is a family of
  identity morphisms, i.e.\ if $P = Q \circ h$.
\item $\Fam(\C)$ is the free set indexed coproduct completion of $\C$;
  that is $\Fam(\C)$ has all set indexed coproducts and there is an
  embedding $\C\to \Fam(\C)$ universal among functors $F:\C\to
  \mathbb{D}$ where $\mathbb{D}$ is a category with set indexed
  coproducts.  Given an $A$-indexed collection of objects
  $(X_a,P_a)_{a\,:A}$ in $\Fam(\C)$, its $A$-indexed coproduct is the
  family $\sum_{a\,: A}(X_a,P_a) = (\sum_{a\,:A}X_a,[P_a]_{a\,:A})$.
\item $\Fam(\C)$ is cocomplete if and only if $\C$ has all small
  connected colimits (Carboni and Johnstone~\cite[dual of Prop.\
  2.1]{JoCa}).
\item $\Fam$ is a functor $\CAT\to\CAT$; given $F:\C\to \mathbb{D}$,
  we get a functor $\Fam(F):\Fam(\C) \to \Fam(\mathbb{D})$ by
  composition: $\Fam(F)(X,P) = (X,F\circ P)$. Here $\CAT$ is the
  category of large categories.
\end{itemize}
\end{remarks}

\noindent When $\C$ is a discrete category, a morphism between families $(X,P)$
and $(Y,Q)$ in $\Fam(\C)$ consists of a function $h:X\to Y$ such that
$P\,x=Q\,(h\,x)$ for all $x$ in $X$.  From a fibrational perspective,
this amounts to the restriction to the split cartesian fragment
$\Fam\,|\C|$ of the fibration $\pi:\Fam(\C) \to \Set$, for $\C$ an
arbitrary category.  This observation is crucial for the
interpretation of $\IR$ codes as functors.  Indeed, given a type $\D$,
which we think of as the discrete category $|\D|$ (with objects terms
of type $D$), we interpret $\IR$ codes as functors $\Fam\,|\D| \to
\Fam\,|\D|$.

\begin{theorem}[$\IR$ functors~\cite{DybSetz2}]
\label{thm:IR-functors}
Let $\D$ be a (possibly large) type. Every code $\gamma : \IR(\D)$
induces a functor
\[
\inter{\gamma}:\Fam\,|\D| \to \Fam\,|\D|
\]
\end{theorem}
\begin{proof}
  We define $\inter{\gamma}:\Fam\,|\D| \to \Fam\,|\D|$ by induction on
  the structure of the code $\gamma$. We first give the action on
  objects:
\begin{align*}
\inter{\iota\,c} (X, P) &= (\one,\lambda\_\,.\,c) \\
\inter{\sigma_A\,f}(X, P)& = \sum_{a\,: A}\inter{f\,a}(X,P) \\
\inter{\delta_A\,F}(X, P)&=\sum_{g\,:A\to X} \inter{F\,(P\circ g)}(X,P)
\end{align*}

We now give the action on morphisms.  Let $(h, \id) : (X, P) \to (Y,
Q)$ be a morphism in $\Fam\,|\D|$, i.e.\ $h : X \to Y$ and $Q \circ h
= P$.
\begin{align*}
\inter{\iota\,c} (h, \id) &= (\id_{\one},\id)\\
\inter{\sigma_A\,f}(h, \id)& = [\inn_a\circ\inter{f\,a}(h, \id)]_{a\,: A}\\
\inter{\delta_A\,F}(h, \id)&=[\inn_{h\circ g}\circ\inter{F(Q \circ h \circ g)}(h, \id)]_{g\,:A\to X}
\end{align*}
Here, the last line type checks $Q \circ h = P$ since $\D$ is discrete. Hence 
\begin{equation}\label{eq:on-the-nose}
Q \circ h \circ g = P \circ g\end{equation}
 and we can apply the induction hypothesis.
\end{proof}

Note how the interpretation of both $\sigma$ and $\delta$ codes makes
essential use of coproducts of families as defined in \cref{rem:Fam}.
In particular, the interpretation of a code $\delta_A F$ uses as the
coproduct's index set the function space $A\!\to\! X$, which is a set
since both $A$ and $X$~are.
 
Ghani et al.~\cite{Malatesta} introduces morphisms between (small)
$\IR$ codes.  The morphisms are chosen to make the interpretation
function $\inter{-}:\IR(\D) \to (\Fam\,|\D| \to \Fam\,|\D|)$ full and
faithful. Thus, transporting composition and identity along this
function makes $\IR(\D)$ into a category, and $\inter{-}:\IR(\D) \to
(\Fam\,|\D| \to \Fam\,|\D|)$ can really be seen as a full and faithful
functor.  We will draw inspiration from this in \cref{sec:pos-ind-rec}
when we generalise the semantics to endofunctors on $\Fam(\C)$ for
possibly non-discrete categories $\C$. Note however that the
definition of morphisms between codes we give here differs from the
one appearing in Ghani et al.~\cite{Malatesta}. The key idea of the
latter is a characterization of the interpretation of $\delta$ codes
as left Kan extensions. In our more general setting where $\C$ can be
a non-discrete category, this characterization fails. As a
consequence, we lose the full and faithfulness of the interpretation
functor $\inter{-}$ and we have to prove by hand that the set of codes
and morphisms between them actually is a category. Full and
faithfulness of the interpretation is convenient and desirable, and
often simplifies calculations. Nonetheless, it is not an essential
property, and we manage to make do without it.

We call a data type \emph{inductive-recursive} if it is the initial
algebra of a functor induced from an $\IR$ code. Let us look at some
examples.

\begin{example}[A universe closed under dependent sums]
\label{ex:sigma}
In the introduction, we introduced a universe in
Equation~\eqref{eq:U-closed-sigma}, containing the natural numbers and
closed under $\Sigma$-types, and claimed that this universe can be
defined via an inductive-recursive definition.  Indeed, one can easily
write down a code $\gamma_{\nattype,\Sigma} :\IR(\Set)$ for a functor
that will have such a universe as its initial algebra:
\[
\gamma_{\nattype,\Sigma} \coloneqq \iota\,\nattype +_{\IR} \delta_{\one}(X \mapsto \delta_{X\ast}(Y \mapsto \iota\,\Sigma (X\ast)\,Y))\,: \IR(\Set)
\]
Here we have used $\gamma +_{\IR} \gamma' \coloneqq \sigma_{\two}\,(0
\mapsto \gamma ; 1 \mapsto \gamma')$ to encode a binary coproduct as a
$\two$-indexed coproduct. Also, in the above, note that $X:\one
\rightarrow \Set$ and so $X \ast$ is simply the application of $X$ to
the canonical element of $\one$. If we decode
$\gamma_{\nattype,\Sigma}$, we get a functor which satisfies
\begin{align*}
       \inter{\gamma_{\nattype,\Sigma}}(U, T)
&\cong (\one + \Sigma u\!:\!U\,.\,T(u) \to U, \inl \_ \mapsto \nattype ; \inr(u, f) \mapsto \Sigma\,x\!:\!T(u)\,.\,T(f(x)))
\end{align*}
so that the initial algebra $(U, T)$ of
$\inter{\gamma_{\nattype,\Sigma}}$, which satisfies $(U, T) \cong
\inter{\gamma_{\nattype,\Sigma}}(U, T)$ by Lambek's Lemma, indeed
satisfies Equation \eqref{eq:U-closed-sigma}.
\end{example}

\begin{example}[A universe closed under dependent function spaces]
\label{ex:pi}
In the same way, we can easily write a down a code for a universe
closed under $\Pi$-types:
\[
\gamma_{\nattype,\Pi} \coloneqq \iota\,\nattype +_{\IR} \delta_{\one}(X \mapsto \delta_{X\ast}(Y \mapsto \iota\,\Pi (X\ast)\,Y))\,: \IR(\Set)
\]
Even though this looks extremely similar to the code in the previous
example, we will see in the next section that there is a big semantic
difference between them.

\end{example}

\section{Positive Inductive-Recursive Definitions}
\label{sec:pos-ind-rec}

\Cref{thm:IR-functors} tells us that $\IR$ codes can be interpreted as
functors on families built over a discrete category. What happens if
we try to interpret $\IR$ codes on the category $\Fam(\C)$, and not
just on the subcategory $\Fam\,|\C|$, whose morphisms are the split
cartesian ones only? Consider the following morphism in $\Fam\,|\C|$:
\[
\xymatrix{
X \ar[dr]_-{P} \ar[rr]^-{h} & & Y \ar[dl]^-{Q} \\
& \C
}
\]
What if the diagram above does not commute on the nose, since $\C$ is
not simply a discrete category, but a category whose intrinsic
structure we want to keep track of? For instance, it is natural to
require that the diagram above only commutes up to isomorphism, i.e.\
$P(x) \cong Q(h(x))$ instead of $P(x) = Q(h(x))$.  What structure is
required to interpret inductive-recursive definitions in this larger
category? The problem is that if we allow for more general morphisms,
we can not prove functoriality of the semantics of a $\delta$ code as
it stands anymore: it is essential to have an actual equality on the
second component of a morphism in $\Fam(\C)$ in order to have a sound
semantics.

In this section we propose a new axiomatization which enables us to
solve this problem. This new theory, which we dub \emph{positive
  inductive-recursive definitions}, abbreviated $\pIR$, represents a
generalization of $\IR$ which allows the interpretation of codes as
functors defined on $\Fam(\C)$ for an arbitrary category $\C$. In
particular, if we choose $\C$ to be a groupoid, i.e.\ a category where
every morphism is an isomorphism, we get triangles commuting up to
isomorphism as morphisms in $\Fam(\C)$.

\subsection{Syntax and Semantics of $\pIR(\C)$}

The crucial insight which guides us when introducing the syntax of
$\pIR$ is to deploy proper functors in the introduction rule of a
$\delta$ code. This enables us to remove the restriction on morphisms
within inductive recursive definitions; indeed, if we know that $F :
(A \to \C)\to \pIR(\C)$ is a \emph{functor}, and not just a
\emph{function}, we do not have to rely on an identity in
\Cref{eq:on-the-nose}, but we can use the second component of a
morphism $(h,k)\,:(X,P)\to (Y,Q)$ in $\Fam(\C)$ to get a map $P\circ g
\to Q\circ h\circ g$; then we can use the fact that $F$ is a functor
to get a morphism between codes $F(P\circ g) \to F(Q\circ h\circ g)$.

But, now we have to roll up our sleeves. For $F: (A \to \C)\to
\pIR(\C)$ to be a functor, we need both $A \to \C$ and $\pIR(\C)$ to
be categories. While it is clear how to make $A \to \C$ a category,
turning $\pIR(\C)$ into a category entails defining both codes and
morphisms between codes simultaneously, in an inductive-inductive
fashion~\cite{NordvallForsberg,Fredthesis}.  We give an axiomatic
presentation of $\pIR$ analogously to the one given in \cref{sec:IR}
for the syntax of $\IR$; however we now have mutual introduction rules
to build both the type of $\pIR(\C)$ codes and the type of $\pIR(\C)$
morphisms, for $\C$ a given category.  The semantics we give then
explains how $\pIR(\C)$ codes can be interpreted as functors on
$\Fam(\C)$, while $\pIR(\C)$ morphisms between such codes can be
interpreted as natural transformations.

\begin{definition}\label{def:pos-IR-codes}
Given a category $\C$ we simultaneously define the type $\pIR(\C)$
of positive inductive-recursive codes on $\C$, and the type of
morphisms between these codes $\Hom{\pIR(\C)}(\_,\_):\pIR(\C)\to
\pIR(\C)\to \type$ as follows:

\begin{itemize}
\item{$\pIR(\C)$ codes:}
\[
\infer{\iota\, c : \pIR(\C)}{c:\C}
\]

\[
\infer{\sigma_{A} f : \pIR(\C)}{A\,:\Set & \quad f\,:A\rightarrow \pIR(\C)}
\]

\[
\infer{\delta_{A} F : \pIR(\C)}{A\,: \Set & \quad F:(A\rightarrow \C)\rightarrow \pIR(\C) 
}
\]

\item{$\pIR(\C)$ morphisms:}
\[
\infer{\cons{\iota}{\iota}(f):\Hom{\pIR(\C)}(\iota\, c,\iota\, c')}{f:\Hom{\C}(c,c')}
\] 

 \[
 \infer{\cons{\sigma}{\sigma}(\alpha, \rho):\Hom{\pIR(\C)}(\sigma_{A}\,f,\sigma_{B}\,g)}{\alpha : A \to B &\quad \rho : \prod_{x:A}\Hom{\pIR(\C)}(f(x),g(\alpha(x)))}
 \]

\[
  \infer{\cons{\delta}{\delta}(\alpha,\rho):\Hom{\pIR(\C)}(\delta_A F,\delta_B G)}{\alpha:B\to A &\quad\rho: \mathsf{Nat}(F,G( - \circ \alpha))}
\]

\end{itemize}

\noindent In the last clause, we have indicated with
$\mathsf{Nat}(F,G( - \circ \alpha))$ the collection of natural
transformations between the functors $F$ and $G( - \circ \alpha) : (A
\to \C) \to \pIR(\C)$. Note also the contravariant twist in the type
of $\alpha : B \to A$ in this clause.
\end{definition} 

We need to make sure that \Cref{def:pos-IR-codes} really defines a
category, i.e. that composition of $\pIR$ morphisms can be
defined, and that it is associative and has identities.
This can be proved by recursion on the structure of morphisms:

\begin{lemma}
\label{lem:Pos-IR-is-a-cat}
Let $\C$ be a category. Then $\pIR(\C)$ is a category with morphisms
given by $\Hom{\pIR(\C)}$.
\end{lemma}
\begin{proof}
  We define $\id^{+}_{x} : \Hom{\pIR(\C)}(x, x)$ by recursion on $x$:
  \begin{align*}
    \id^{+}_{\iota c}    &= \cons{\iota}{\iota}(\id_{c}) \\
    \id^{+}_{\sigma_A\,f} &= \cons{\sigma}{\sigma}(\id_A, \lambda\,a\,.\, \id^{+}_{f(a)}) \\
    \id^{+}_{\delta_A F}  &= \cons{\delta}{\delta}(\id_A, \lambda\,h\,.\,  \id^{+}_{F(h)})
  \end{align*}
  Composition $\_\circ_{\pIR}\_ : \Hom{\pIR(\C)}(y, z) \to
  \Hom{\pIR(\C)}(x, y) \to \Hom{\pIR(\C)}(x, z)$ is defined by
  recursion on $f : \Hom{\pIR(\C)}(y, z)$ and $g : \Hom{\pIR(\C)}(x,
  y)$:
  \begin{align*}
    \cons{\iota}{\iota}(f) \circ_{\pIR} \cons{\iota}{\iota}(g) &= \cons{\iota}{\iota}(f \circ_{\C} g) \\
    \cons{\sigma}{\sigma}(\alpha, \rho) \circ_{\pIR} \cons{\sigma}{\sigma}(\beta, \tau) &= \cons{\sigma}{\sigma}(\alpha \circ \beta, \lambda\,x\,.\,\rho(\alpha(x)) \circ_{\pIR} \tau(x)) \\
    \cons{\delta}{\delta}(\alpha, \rho) \circ_{\pIR} \cons{\delta}{\delta}(\beta, \tau) &= \cons{\delta}{\delta}(\beta \circ \alpha, \lambda\,h\,.\,\rho(h \circ \beta) \circ_{\pIR} \tau(h))
  \end{align*}
  Three more straightforward inductions prove that composition is
  associative, and that $\id^{+}$ is both a left and a right unit for
  composition.
\end{proof}

We now explain how each code $\gamma : \pIR(\C)$ is interpreted as an
endofunctor
\[
\inter{\gamma}:\Fam(\C) \to \Fam(\C)
\]
We call a functor which is isomorphic to a functor induced by an
$\pIR$ code an $\pIR$ functor. The semantics of $\pIR$ closely follows
the one given in \Cref{sec:IR}; as before we make essential use of
coproducts in $\Fam(\C)$. Having said that, the crucial feature which
separates the semantics of $\pIR$ from the semantics of $\IR$ is the
following: when explaining the semantics of $\IR$, we first interpret
$\IR$ codes as functors and then later define morphisms between codes.
We can then interpret the morphisms as natural transformations between
the corresponding functors. In $\pIR$, the type of codes and the type
of morphisms between codes are simultaneously defined in an
inductive-inductive way, and therefore they are also decoded
simultaneously as functors and natural transformations
respectively. This is exactly what the elimination principle for an
inductive-inductive definition gives.

In the following theorem, note that there is no restriction on the
category $\C$ --- all structure that we need comes for free from the
families construction $\Fam$.

\begin{theorem}[$\pIR$ functors]
\label{thm:p-IR-functor}
Let $\C$ be an arbitrary category.
\begin{enumerate}[label=(\roman*)]
\item Every code $\gamma\,:\pIR(\C)$ induces a functor
  $\inter{\gamma}:\Fam(\C) \to \Fam(\C)$.
\item Every morphism $\rho:\pIR(\C)(\gamma, \gamma')$ for codes
  $\gamma,\gamma'\,:\pIR(\C)$ gives rise to a natural transformation
  $\inter{\rho}:\inter{\gamma}\stackrel{\cdot}{\longrightarrow}\inter{\gamma'}$.
\end{enumerate}
\end{theorem}
\proof
  While the action on objects is the same for both $\pIR$ and $\IR$
  functors, the action on morphisms is different when interpreting a
  code of type $\delta_A F$: in the semantics of $\pIR$ we exploit the
  fact that $F:(A\rightarrow \C)\rightarrow \pIR(\C)$ is now a
  functor, so that it also has an action on morphisms (which we, for
  the sake of clarity, write $F_{\to}$). We give the action of $\pIR$
  functors on morphisms only, and refer to the semantics given in
  \cref{thm:IR-functors} for the action on objects of $\Fam(\C)$.

  The action on morphisms is given as follows.  Let $(h, k) : (X,P)
  \to (Y,Q)$ in $\Fam(\C)$.  We define $\inter{\gamma}(h, k) :
  \inter{\gamma}(X, P) \to \inter{\gamma}(Y, Q)$ by recursion on
  $\gamma$:
\begin{align*}
\inter{\iota\, c} (h, k) &= (\id_{\one},\id_c)\\
\inter{\sigma_A f}(h, k)& = [\inn_a\circ\inter{f\,a}(h, k)]_{a\,: A}\\
\inter{\delta_A F}(h, k)&=[\inn_{h\circ g}\circ\inter{F(Q\circ h\circ g)}(h, k)\circ \inter{F_{\to}(g^{*}(k))}_{(X,P)}]_{g\,:A\to X}
\end{align*}
In the last clause $g^{*}(k) \,:P\circ
g\stackrel{\cdot}{\longrightarrow} Q\circ h\circ g$ is the natural
transformation with component $g^{*}(k)_{a} = k_{g\,a} : P(g\,a) \to
Q(k(g\,a))$; note that such a natural transformation is nothing but
the vertical morphism above $A$ obtained by reindexing $(\id_X,k)$
along $g$ in the families fibration $\pi : \Fam(\C) \to \Set$.

We now explain how an $\pIR$ morphism $\rho : \gamma \to \gamma'$ is
interpreted as a natural transformation $\inter{\rho} : \inter{\gamma}
\stackrel{\cdot}{\longrightarrow} \inter{\gamma'}$ between $\pIR$
functors by specifying the component $\inter{\rho}_{(X, P)}$ at
$(X,P):\Fam(\C)$. Naturality of these transformations can be proved by
a routine diagram chase.
\begin{align*}
\inter{\cons{\iota}{\iota}(f)}_{(X,P)}        & =(\id_\one, f) \\
\inter{\cons{\sigma}{\sigma}(\alpha, \rho)}_{(X,P)}    & = [\inn_{\alpha(x)} \circ \inter{\rho(x)}_{(X,P)}]_{x\,:A}                           \\
\inter{\cons{\delta}{\delta}(\alpha,\rho)}_{(X,P)} & =[\inn_{g\circ
    \alpha}\circ \inter{\rho_{(P\circ g)}}_{(X,P)}]_{g:A\to
  X}\rlap{\hbox to 98 pt{\hfill\qEd}}
\end{align*}

\begin{remark}
  In the conference version of this paper~\cite{calco}, we considered
  a different collection of morphisms; since more morphisms makes it
  easier to define codes, we tried to include as many morphisms as
  possible. As a result, the proof that $\pIR(\C)$ is a category
  becomes quite long and tedious, although straightforward. In this
  presentation, we have instead decided to restrict ourselves to the
  smallest possible ``usable'' combination of morphisms. It should be
  noted that our results are completely parametric in the choice of
  morphisms used; any collection that represents natural
  transformations between the codes works, as long as the identity
  morphisms and composition can be defined. The range spans all the
  way from no non-identity morphisms at all (in which case it is
  rather hard to define a functor $(A \to \C)\to \pIR(\C)$!) to taking
  $\Hom{\pIR}(x, y) = \inter{x} \to \inter{y}$, which gives rise to a
  full and faithful interpretation by definition. The latter would
  mean that the interpretation $\inter{-}$ would need to be defined
  simultaneously with the codes, with the effect that the very
  definition of positive inductive-recursive definitions itself would
  be inductive-recursive. To avoid this stronger assumption in the
  metatheory, we prefer the current formulation, where the meta-theory
  only uses inductive-inductive definitions --- a much weaker
  principle.
\end{remark}

Let us now return to the examples from the end of \cref{sec:IR}.

\begin{example}[A universe closed under dependent sums in $\Fam(\Dual{\Set})$]
\label{ex:sigma-again}
In \cref{ex:sigma}, we defined an ordinary $\IR$ code
$\gamma_{\nattype,\Sigma} : \IR(\Set)$ for a universe closed under
sigma types.  We can extend this code to an $\pIR$ code
\[
\gamma_{\nattype,\Sigma} = \iota\,\nattype +_{\IR} \delta_{\one}(X \mapsto \delta_{X\ast}(Y \mapsto \iota\,\Sigma (X\ast)\,Y)) : \pIR(\Dual{\Set})
\]
where now $G \coloneqq Y \mapsto \iota\,\Sigma (X\ast)\,Y$ and $F
\coloneqq X \mapsto \delta_{X\ast}\,G$ need to be functors.  Given $f
: Y \to Y'$ in $X \to \Dual{\Set}$, i.e. an $X$-indexed collection of
morphisms \ $f_x : Y(x) \to Y'(x)$ in $\Dual{\Set}$, we have $\Sigma x
: (X\ast) . f_x : \Sigma (X\ast)\,Y \to\,\Sigma (X\ast)\,Y'$ in
$\Dual{\Set}$ so that we can define
\[
G(f) : \iota\,\Sigma (X\ast)\,Y \to \iota\,\Sigma (X\ast)\,Y'
\]
by $G(f) = \cons{\iota}{\iota}(\Sigma x : (X\ast) . f_x)$.

We also need $F$ to be a functor. Given $f : X \to X'$ in $\one \to
\Dual{\Set}$, we need to define $F(f) : \delta_{X\ast}\,G \to
\delta_{X'\ast}\,G$.  According to \cref{def:pos-IR-codes}, it is
enough to give an $\alpha : X' \ast \to X \ast$ and a natural
transformation $\rho$ from $G$ to $G (- \circ \alpha)$. We can choose
$\alpha = f_{\ast}$, and $\rho$ to be the natural transformation whose
component at $Y : X\ast \to \Dual{\Set}$ is given by $\rho_Y=
\cons{\iota}{\iota}([\inn_{f_{\ast} x}]_{x : X'\ast})$, where
$[\inn_{f_{\ast} x}]_{x : X'\ast}\,: \Sigma (X'\ast)\,Y\circ f_\ast
\to \Sigma (X\ast)\,Y$. Notice that working in $\Dual{\Set}$ made sure
that $f_{\ast}$ was going in the right direction.
\end{example}

\begin{example}[A universe closed under dependent function spaces in $\Fam(\Set^{\cong})$]
\label{ex:pi-again}
In \cref{ex:pi}, we saw how we could use induction-recursion to define
a universe closed under $\Pi$-types in $\Fam\,|\Set|$, using the
following code:
\[
\gamma_{\nattype,\Pi} = \iota\,\nattype +_{\IR} \delta_{\one}(X \mapsto \delta_{X\ast}(Y \mapsto \iota\,\Pi (X\ast)\,Y))\,: \IR(\Set)
\]
If we try to extend this to an $\pIR$ code in $\Fam(\Set)$ or
$\Fam(\Dual{\Set})$, we run into problems. Basically, given a morphism
$f : X' \to X$, we need to construct a morphism $\Pi\,X'\,(Y \circ f)
\to \Pi\,X\,Y$, which of course is impossible if e.g.\ $X' = \zero$,
$X = \one$, and $Y \ast = \zero$.

Hence the inherent contravariance in the $\Pi$-type means that
$\gamma_{\nattype,\Pi}$ does not extend to a $\pIR(\Set)$ or
$\pIR(\Dual{\Set})$ code.  However, if we move to the groupoid
$\Set^{\cong}$, which is the subcategory of $\Set$ with only
isomorphisms as morphisms, we do get an $\pIR(\Set^{\cong})$ code
describing the universe in question, which is still living in a
category beyond the strict category $\Fam\,|\Set|$. It would be
interesting to understand the relevance of positive
induction-recursion to Homotopy Type Theory~\cite{hottbook} where
groupoids and their higher order relatives play such a prominent role.
\end{example}

\section{Stronger elimination principles}
\label{sec:elim-principle}

From \Cref{ex:sigma-again} we know that the $\IR$ code
$\gamma_{\nattype,\Sigma}$ defining a universe containing the set of
natural numbers $\nattype$ and closed under $\Sigma$-type can be
extended to a $\pIR$ code of type $\pIR(\Set^{\cong})$ or
$\pIR(\Dual{\Set})$. Thus, the code $\gamma_{\nattype,\Sigma}$ can be
interpreted as an endofunctor on $\Fam(\Set^{\cong})$ or on
$\Fam(\Dual{\Set})$ respectively.
In this section we aim to explore by means of an example what the
elimination principle for $\pIR$ codes can be used for: we show how
the simple elaboration of the code $\gamma_{\nattype,\Sigma}$ to a
code of type $\pIR(\Set^{\cong})$ offers us the possibility to
implement a more sophisticated recursion principle on the universe we
are currently building.

Recall that from the perspective of initial algebra semantics, the
elimination principle for a type is captured by the universal property
of the initial algebra: if $F$ is an endofunctor and $(\mu_F,in_F)$
its initial algebra, then we know that for any other algebra $(X,f)$
there exists a (unique) $F$-algebra homomorphism $\alpha_f\,:\mu_F\to
X$ which makes the following diagram commute:
 \[ \xymatrix{
   F \mu_F \ar[r]^{in} \ar[d]_{F(\alpha_g)} & \mu_F \ar[d]^{\alpha_g}\\
   F X \ar[r]_{g} & X}
 \] 
 The initial property of $(\mu_F,in_F)$ thus gives us a definition by
 recursion into any other type possessing the right $F$-algebra
 structure. By working in $\Fam(\C)$ instead of $\Fam\,|\C|$, we are
 allowing many more algebras compared to ordinary inductive-recursive
 definitions, or put differently, we get a stronger elimination
 principle.
\begin{example}
  To see why a stronger elimination principle is sometimes necessary,
  consider the initial algebra $((U^*,T^*),(\inn_0,\inn_1))$ for a
  code $\gamma_{\one,\nattype,\Sigma}\,:\pIR(\Set^{\cong})$
  representing a universe containing a set $\one$ with only one
  element, the set $\nattype$ of natural numbers and moreover closed
  under $\Sigma$-types. The universe $U^*$ contains many codes for
  ``the same'' set, up to isomorphism. For instance, it contains codes
  for each of the following isomorphic sets:
  \begin{align*}
    \one\cong&(\Sigma\one)\one\cong(\Sigma\one)(\Sigma\one)\one\,\,\,\ldots\\
    \nattype\cong&(\Sigma\nattype)\one\cong(\Sigma\one)\nattype\cong(\Sigma\one)(\Sigma
    \one)\nattype\,\,\,\ldots
  \end{align*}
  Moreover, for each $\Sigma$-type the following isomorphism holds:
  \begin{equation}\label{eq:Sigma-iso}
    (\Sigma z:(\Sigma x:A)B(x))C(z)\cong
    (\Sigma x:A)(\Sigma y\,:B(a))C(\pair{x}{y})
  \end{equation}
  Therefore, for each $\Sigma$ set with at least two nested
  $\Sigma$'s$, U^*$ contains a code for both these ways to
  parenthesize a $\Sigma$-type.  It might be advantageous to instead
  keep a single representative for each isomorphism class. We might
  hope to do so using the initiality of $(U^*,T^*)$, and indeed, the
  elimination principle for positive inductive-recursive definitions
  allows us to do exactly that.

  First of all we need to decide what normal forms for elements in the
  universe we want. We can specify this by defining a predicate
  $\NF\,:U^*\to \Set$ on the universe $(U^*,T^*)$, which decides if a
  set is in normal form: we decree that the codes for the sets $\one$
  and $\nattype$ are in normal form, and a code for $\Sigma\,A\,B$ is
  in normal form if $A$ is in normal form, $B(a)$ is in normal form
  for each $a : A$, $A$ is not $\one$, and finally it is of the form
  of the right hand side of \eqref{eq:Sigma-iso}. There is of course
  some room for different choices here. Formally, and employing some
  cleverness in how we set things up, we can define the predicate by
  the elimination principle for $U^*$ by the following clauses:
  \begin{align*}
    \NF(\hat{\one}) &= \top \\
    \NF(\hat{\nattype}) &= \top \\
    \NF(\hat{\Sigma}\, \hat{\one}\, b) &= \bot \\
    \NF(\hat{\Sigma}\,\hat{\nattype}\, b) &= \forall n : \nattype\,.\, \NF(b(n)) \\
    \NF((\hat{\Sigma}\, (\hat{\Sigma}\, a'\, b')\, b) &= \bot
  \end{align*}

  We now define a new family $(U_{\NF},T_{\NF})$, containing sets in
  normal forms only, by letting
  \begin{align*}
    U_\NF:=&\,(\Sigma u:U^*)\NF(u)\\
    T_\NF(u,p):=&\,T^*(u)\end{align*} We can also define a
  $\Fam(\Set^{\cong})$ morphism
  $(\phi,\eta)\,:\inter{\gamma_{\one,\nattype,\Sigma}}(U_\NF,T_\NF)\to(U_\NF,T_\NF)$
  which endows $(U_{\NF},T_{\NF})$ with an
  $\inter{\gamma_{\one,\nattype,\Sigma}}$-algebra structure. For this,
  is it crucial that we are working in $\Fam(\Set^{\cong})$ and not
  $\Fam\,|\Set|$, since we can only expect that a $\Sigma$-type of
  normal forms is isomorphic to a normal form, not equal to one; i.e.\
  if $A$ is in normal form, and $B(a)$ is in normal form for all $a :
  A$, then $\Sigma\,A\,B$ is not necessary normal (as e.g.\ $A = \one$
  shows), but we can always find a normal form isomorphic to
  $\Sigma\,A\,B$. The function $\phi$ maps $A$ and $B$ to this normal
  form, and $\eta$ is a proof that it is indeed isomorphic to
  $\Sigma\,A\,B$. We only give the definition of $\phi :
  \inter{\gamma_{\one,\nattype,\Sigma}}_0(U_\NF,T_\NF)\to U_\NF$ here;
  the definition of $\eta$ follows the same pattern.\enlargethispage{2\baselineskip}
  \begin{align*}
    \phi(\hat{1}) & = (\hat{\one} , \ast) \\
    \phi(\hat{\nattype}) &= (\hat{\nattype} , \ast) \\
    \phi(\hat{\Sigma}\,(\hat{\one} , p)\,b) &= (\pi_0 b(\ast) , \pi_1 b(\ast)) \\
    \phi (\hat{\Sigma}\,(\hat{\nattype} , p)\,b) &= (\hat{\Sigma}\,\hat{\nattype}\,(\pi_0 \circ b) , n \mapsto \pi_1(b(n))) \\
\phi (\hat{\Sigma}\,(\hat{\Sigma}\,\hat{\nattype}\, b', p)\, b) &=
(\hat{\Sigma}\, \hat{\nattype}\,(n \mapsto  \pi_0(\phi(\ldots))) , (n
\mapsto \pi_1(\phi(\ldots)))) \\
\text{where } &\phi(\ldots) = \phi(\hat{\Sigma}\, (b'(n) , p(n)) (y \mapsto b(n, y))) \\
\phi(\hat{\Sigma}\,(\hat{\Sigma}\,(\hat{\Sigma}\,a\,b)\, b' , p)\,c) & \quad \text{impossible case by the def.\ of $\NF$; we have $p : \bot$} \\
\phi(\hat{\Sigma}\,(\hat{\Sigma}\,\hat{\one}\, b' , p)\,c) & \quad \text{impossible case by the def.\ of $\NF$; we have $p : \bot$}
  \end{align*}
  By initiality of $(U^*,T^*)$ we get a morphism $(\No,\Co)$ making
  the following diagram commute:
  \[ \xymatrix{
    \inter{\gamma_{\one,\nattype,\Sigma}}(U^*,T^*)  \ar[r]^{\,\,\quad(\inn_0,\inn_1)} \ar[d]_{ \inter{\gamma_{\one,\nattype,\Sigma}}(\No,\Co)} & (U^*,T^*) \ar[d]^{(\No,\Co)}\\
    \inter{\gamma_{\one,\nattype,\Sigma}}(U_{\NF},T_{\NF})
    \ar[r]_{\,\,\quad(\phi,\eta)} & (U_{\NF},T_{\NF})}
  \]
  The map $(\No,\Co)$ recursively computes the normal form for each
  set in the universe $(U^*,T^*)$. Indeed, $\No:U^*\to U_\NF$ maps
  each name $u$ of a set $T(u)$ in the universe to the name of the
  corresponding set in normal form, while the natural transformation
  $\Co_u:T^*(u)\cong T_\NF(\No(u))$ ensures that the code actually
  denotes isomorphic sets. Of course, we do not get $(\No,\Co)$ for
  free; defining $\phi$ and $\eta$ already amounts to most of the work
  for the full definition. The point is rather that initiality in
  $\Fam(\Set^{\cong})$ is a definitional principle which \emph{allows
    us} to define $\No$ and $\Co$. Furthermore, by using initiality,
  we can give a structured definition, where we only have to consider
  the separate cases in isolation.
\end{example}

\begin{example}
  As another example of the use of elimination principles beyond
  ordinary inductive-recursive definitions, we can define functions
  between universes with different ground sets. Consider two universes
  $U_1$, $U_2$ closed under the same type-theoretic operations, but
  containing different ground sets $B_1$, $B_2$.  Given a function
  $B_1 \to B_2$, we would like to be able to extend this function to a
  function $U_1 \to U_2$ between all of the two universes. For
  example, we could have a universe
  $(U_{\nattype,\Sigma},T_{\nattype,\Sigma})$, closed under
  $\Sigma$-types and containing the natural numbers $\nattype$, and
  another universe $(U_{\mathbb{Z},\Sigma},T_{\mathbb{Z},\Sigma})$
  also closed under $\Sigma$-types but instead containing the integers
  $\mathbb{Z}$ as ground set. There ought to exist a function between
  them in $\Fam(\Dual{\Set})$ (the contravariance is needed for the
  negative occurrence of $U$ in the code for the sigma type), since
  clearly these two universes are closely related. By the elimination
  principle for positive inductive-recursive definitions, it suffices
  to provide a function between the ground sets, i.e.\ a function from
  $\mathbb{Z}$ into $\nattype$, for instance the absolute value
  function or the square function. In detail, every function
  $f\,:\mathbb{Z}\to \nattype$ induces a $\Fam(\Dual{\Set})$-morphism
 $$\inter{\gamma_{\nattype,\Sigma}}(U_{\mathbb{Z},\Sigma},T_{\mathbb{Z},\Sigma})\longrightarrow(U_{\mathbb{Z},\Sigma},T_{\mathbb{Z},\Sigma})$$ 
 showing that $(U_{\mathbb{Z},\Sigma},T_{\mathbb{Z},\Sigma})$ has an
 $\inter{\gamma_{\nattype,\Sigma}}$-algebra structure. Therefore,
 initiality of $(U_{\nattype,\Sigma},T_{\nattype,\Sigma})$ gives us a
 map
 $(U_{\nattype,\Sigma},T_{\nattype,\Sigma})\to(U_{\mathbb{Z},\Sigma},T_{\mathbb{Z},\Sigma})$
 which uses $f$ to recursively compute the embedding of
 $(U_{\nattype,\Sigma},T_{\nattype,\Sigma})$ into
 $(U_{\mathbb{Z},\Sigma},T_{\mathbb{Z},\Sigma})$.
\end{example}

\section{Application: A Concrete Representation of Nested Types}
\label{sec:nested-types}

Nested data types~\cite{amu05} have been used to implement a number of
advanced data types in languages which support higher-kinded types,
such as the widely-used functional programming language Haskell. Among
these data types are those with constraints, such as perfect
trees~\cite{hin00}; types with variable binding, such as untyped
$\lambda$-terms~\cite{FTP99}; cyclic data structures~\cite{ghuv06};
and certain dependent types~\cite{mm04}. 

A canonical example of a nested data type is $\mathtt{Lam} : \Set \to \Set$
defined in Haskell as follows:
\begin{verbatim}
    data Lam a = Var a | App (Lam a) (Lam a) | Abs (Lam (Maybe a)) 
\end{verbatim} 
\noindent 
The type \verb|Lam a| is the type of untyped $\lambda$-terms over
variables of type \verb|a| up to $\alpha$-equivalence. Here, the
constructor \verb|Abs| models the bound variable in an abstraction of
type \verb|Lam a| by the \verb|Nothing| constructor of type
\verb|Maybe a|, and any free variable \verb|x| of type \verb|a| in an
abstraction of type \verb|Lam a| by the term \verb|Just x| of type
\verb|Maybe a|; The key observation about the type \verb|Lam a| is
that elements of the type \verb|Lam (Maybe a)| are needed to build
elements of \verb|Lam a| so that, in effect, the entire family of
types determined by \verb|Lam| has to be constructed
simultaneously. Thus, rather than defining a family of inductive
types, the type constructor \verb|Lam| defines a {\em type-indexed
  inductive family of types}.  The kind of recursion captured by
nested types is a special case of {\em non-uniform
  recursion}~\cite{bla00}.

On the other hand, ordinary non-nested data types such as
\verb|List a| or \verb|Tree a| can be represented as
\emph{containers}~\cite{AGT,Abb}. Recall that a container $(S, P)$ is
given by a set $S$ of shapes, together with a family $P : S \to \Set$
of positions. Each container gives rise to a functor $\interC{S, P} :
\Set \to \Set$ defined by $\interC{S, P}(X) = \Sigma s : S\,.\,P(s)
\to X$. Since also nested data types such as \verb|Lam| have type
$\Set \to \Set$, it make sense to ask the following question:
\emph{Are nested data types representable as containers?}  There would
be benefits of a positive answer, since container technology could
then be applied to nested data types. For instance, we could operate
on nested types using container operations such as the derivative, and
classify the natural transformations between them. Note in particular
that the canonical recursion operator {\tt fold} for nested data types
is a natural transformation.

We give a positive answer to the above question using $\pIR$. As far
as we are aware, this is a new result.  We sketch our overall
development as follows:
\begin{enumerate}[label=(\roman*)]
\item We define a grammar $\Nest$ for defining nested types and a
  decoding function $\Int{-}:\Nest \rightarrow (\Set \to \Set)
  \rightarrow (\Set \to \Set)$. The data types we are interested in
  arise as initial algebras $\mu \Int{N}$ for elements $N$ of the
  grammar.

\item We show that $\Int{N}$ restricts to an endofunctor
  $\Int{N}_{\cont} : \cont \to \cont$ on the category $\cont$ of
  containers.

\item Noting that $\cont = \Fam(\Dual{\Set})$, we use $\pIR$ to define
  $\Int{N}_{\cont}$. Hence by the results of this paper,
  $\Int{N}_{\cont}$ has an initial algebra $\mu \Int{N}_{\cont}$. We
  finish by arguing that $\mu \Int{N} = \interC{\mu \Int{N}_{\cont}}$
  and hence that, indeed, nested types are containers.
\end{enumerate}%

\paragraph{\bf A Grammar for Nested Types}
We now present a grammar for defining nested data types.  Since our
point is not to push the theory of nested data types, but rather to
illustrate an application of positive induction-recursion, we keep the
grammar simple. The grammar we use is
\[
\calF = \Id \; | \; K\,C \; | \; \calF + \calF \; | \; \calF \times \calF \; | \; \calF \circledast \calF
\]
where $C$ is any container. The intention is that $\Id$ stands for the
identity functor mapping a functor to itself, $K C$ stands for the
constant functor mapping any functor to the interpretation of the
container $C$, $+$ and $\times$ stand for the coproduct and product of
functors respectively, and $\circledast$ for the pointwise composition
of functors. These intentions are formalised by a semantics for the
elements of our grammar given as follows
\[
\begin{array}{lll}
\Int{-} & : & \Nest \rightarrow (\Set \to \Set) \rightarrow (\Set \to \Set)\\
\Int{\Id} \; F & = & F \\
\Int{K \; C} \; F & = & \interC{C}  \\
\Int{\calF_0 + \calF_1} \; F & = & \Int{\calF_0} \; F  + \Int{\calF_1} \; F  \\ 
\Int{\calF_0 \times \calF_1} \; F & = & \Int{\calF_0} \; F  \times \Int{\calF_1} \; F  \\ 
\Int{\calF_0 \circledast \calF_1} \; F & = & \Int{\calF_0} \; F \circ \Int{\calF_1} \; F  \\ 
\end{array} 
\]
For example, the functor 
\[
L \; F \; X \;\; = \;\; X \; +\; (FX \times FX) \; +\; F(X+1)
\]
whose initial algebra is the type \verb|Lam|, is of the form $\Int{N_L}$ with 
\[
N_L = K\,I_C \; + \; (\Id \times \Id) \; + \; (\Id \circledast (K\, M))
\]
where $I_C =(\one, \_ \mapsto \one)$ is the container with one shape
and one position, representing the identity functor on $\Set$, and $ M
= (\two , x \mapsto \mathsf{if} \; x \; \mathsf{then} \; \one
\;\mathsf{else} \; \zero) $ is the container with two shapes, the
first one with one position, and the second one with no position. $M$
represents the functor on $\Set$ mapping $X$ to $X+\one$.

\paragraph{\bf Nested Types as Functors on Containers}
The next thing on our agenda is to show that every element $N$ of
$\Nest$ has an interpretation as an operator on containers
$\Int{N}_{\cont} : \cont \rightarrow \cont$, such that
$\Int{N}_{\cont}$ is the restriction of $\Int{N}$ to the subcategory
of functors that are extension of containers. This is done easily
enough by recursion on $N$, noting that containers are closed under
coproduct, product and composition:

\begin{lemma}[\cite{AGT}]
  Let $(S, P)$ and $(S', P')$ be containers. Define
  \begin{align*}
  (S, P) + (S', P') &\coloneqq (S + S', [P, P']) \\
  (S, P) \times (S', P') &\coloneqq (S \times S', (s, s') \mapsto P(s) + P'(s')) \\
  (S, P) \circ (S', P') &\coloneqq (\Sigma s : S\,.\,(P(s)\to S'), (s, f) \mapsto \Sigma p : P(s)\,.\,P'(f(p)))
  \end{align*}
We then have
  \begin{align*}
  \interC{(S, P) + (S', P')} &\cong \interC{S, P} + \interC{S', P'} \\
  \interC{(S, P) \times (S', P')} &\cong \interC{S, P} \times \interC{S', P'} \\
  \interC{(S, P) \circ (S', P')} &\cong \interC{S, P} \circ
  \interC{S', P'}\rlap{\hbox to 103.5 pt{\hfill\qEd}}
  \end{align*}
\end{lemma}\medskip

\noindent Thus, the interpretation $\Int{N}$ indeed restricts to the subcategory
$\cont$:

\begin{proposition}
Define $\Int{-}_{\cont}  : \Nest \rightarrow \cont \rightarrow \cont$ by
\[
\begin{array}{lll}
\Int{\Id}_{\cont} \; C & = & C \\
\Int{K \; (S, P)}_{\cont} \; C & = & (S, P)  \\
\Int{\calF_0 + \calF_1}_{\cont} \; C & = & \Int{\calF_0}_{\cont} \; C  + \Int{\calF_1}_{\cont} \; C  \\ 
\Int{\calF_0 \times \calF_1}_{\cont} \; C & = & \Int{\calF_0}_{\cont} \; C  \times \Int{\calF_1}_{\cont} \; C  \\ 
\Int{\calF_0 \circledast \calF_1}_{\cont} \; C & = & \Int{\calF_0}_{\cont} \; C \circ \Int{\calF_1}_{\cont} \; C \\
\end{array} 
\]
The following diagram then commutes:
\begin{equation*}
\label{eq:nest-as-cont}
\begin{gathered}
\xymatrix{
\cont \ar[r]^-{\interC{-}} \ar[d]_{\Int{N}_\cont} & (\Set \to \Set) \ar[d]^{\Int{N}} \\
\cont \ar[r]_-{\interC{-}}        & (\Set \to \Set)\rlap{\hbox to 151,5 pt{\hfill\qEd}}
} 
\end{gathered}
\end{equation*}
\end{proposition}

\noindent Coming back to our running example, if we consider the nested code
$N_L$ for \verb|Lam|, we have $\Int{N_L}_{\cont}(S,P) = (S_L,P_L)$
with
\[
\begin{array}{lll}
S_L & = & 1 + (S \times S) + \Sigma s:S. \; P(s) \rightarrow 2 \\
P_L \; (\inn_1\,\ast) & = & 1 \\
P_L \; (\inn_2\,(s,s')) & = & P(s) + P(s') \\
P_L \; (\inn_3\,(s,f)) & = & \Sigma p:P(s). \; \mathsf{if} \; f(p) \; \mathsf{then} \; \one \;\mathsf{else} \; \zero
\end{array}
\]
We see that indeed the positions $P$ show up in the equation for the
shape $S_L$, so that we should expect an inductive-recursive
definition to be the initial solution to this set of equations.

\paragraph{\bf Nested Types are Containers}
We know that $\cont=\Fam(\Dual{\Set})$. Now, we want to show that for
every code $N:\Nest$, the functor $\Int{N}_{\cont}$ is a $\pIR$
functor: to see this one needs to carefully examine the constructions
on families used to build $\Int{N}_{\cont}$. We will need to show that
we can emulate the identity functor, containers and container product,
coproduct and composition using $\pIR$ codes. Most of these are
straightforward, but container composition will require some
sophistication: we will need to observe that all $\pIR$ codes in
question in fact are of a particularly simple, \emph{uniform} form. We
deal with each code in the nested grammar in turn.

\begin{lemma}[$\pIR$ codes for $\Id$ and $K (S, P)$] \label{thm:nestIdK}
\mbox{}
  \begin{enumerate}[label=(\roman*)]
  \item $\inter{\delta_{\one}((X : \one \to \Dual{\Set}) \mapsto \iota(X \ast))}C \cong C$.
  \item $\inter{\sigma_{S}(s \mapsto \iota\,P(s))}C \cong (S, P)$. \qed
  \end{enumerate}
\end{lemma}

For encoding container coproducts, we can reuse the binary coproducts
$+_{\IR}$ on $\pIR$ codes from \cref{ex:sigma}.

\begin{lemma}[$\pIR$ codes for $N + N'$]
\label{thm:nestPlus}
$\inter{\gamma +_{\IR} \gamma'}C \cong
 (\inter{\gamma}_0C + \inter{\gamma'}_0C, [\;\inter{\gamma}_1C, \inter{\gamma'}_1C\;])$. \qed
\end{lemma}

Emulating products of containers requires a little more work. The
basic idea is that we get the product of two codes $\gamma$ and
$\gamma'$ by replacing all occurrences of the terminating code
$\iota\, c$ in the first code $\gamma$ by the second code $\gamma'$,
where, in turn, we replace all codes $\iota\,c'$ with $\iota\,(c
\times c')$. In general, we can replace $\iota\,c'$ with $\iota\,G(c,
c')$ for a functor $G : \C \times \C \to \C$. Formally, we define a
functor $\_[\iota x \longmapsto \iota G(\_, x)] : \pIR \times \C \to
\pIR$ for such a functor $G$ by
\begin{align*}
  (\iota c')[\iota\,x\longmapsto \iota G(c,x)] &= \iota\,G(c, c') \\
  (\sigma_A\, f)[\iota\,x\longmapsto G(c,x)] &= \sigma_A\,(\lambda a\,.\,f(a)[\iota\,x\longmapsto \iota G(c,x)]) \\
  (\delta_A\, F)[\iota\,x\longmapsto G(c,x)] &= \delta_A\,(\lambda h\,.\,F(h)[\iota\,x\longmapsto \iota G(c,x)])
\end{align*}
See the formal development~\cite{agda} for the action on morphisms,
which needs to be defined simultaneously in order to show that
$F(h)[\iota\,x\longmapsto \iota G(c,x)]$ in the $\delta$ case again is a
functor. Using this, we can now define the product $\gamma \times_G
\gamma'$ of two codes with respect to the functor $G$:
\begin{align*}
   (\iota\,c) \times_G \gamma &= \gamma[\iota\,x\longmapsto \iota G(c, x)] \\
   (\sigma_A\,f) \times_G \gamma &= \sigma_A (\lambda a\,.\, f(a) \times_G \gamma) \\
   (\delta_A\,F) \times_G \gamma &= \delta_A (\lambda h\,.\, F(h) \times_G \gamma)
\end{align*}
Again, we need to simultaneosuly show that $\times_G$ is functorial in
order for $F(h) \times \gamma$ in the $\delta$ case to be a functor.

\begin{lemma}
\label{thm:nestTimes}
  $\inter{\gamma \times_G \gamma'}C \cong (\inter{\gamma}_0 C \times \inter{\gamma'}_0 C, (s, s') \mapsto G(\inter{\gamma}_1 C s, \inter{\gamma}_1 C s'))$. \qed
\end{lemma}
In particular, if we choose $G(X, Y) = X + Y$, we recover the
container product.

Finally, we get to container composition. Composition of $\pIR$ codes
(and $\IR$ codes) is an open problem in general, but since we are
interested in emulating composition of containers, one could hope that
there is more structure to be exploited, and this is indeed the
case. The main insight is that all codes in the image of the
translation are \emph{uniform}, in the sense of unpublished work by
Peter Hancock. Intuitively, a $\pIR$ code is uniform if the
\emph{shape} of the code (i.e. $\sigma$/$\delta$/$\iota$ followed by
$\sigma$/$\delta$/$\iota$ followed by\ldots) is independent of the
arguments; e.g. $\sigma_A\,(\lambda x\,.\,\sigma_{B(x)}\,(\lambda
y\,.\,\delta_{C(x, y)}\,(\lambda z\,.\,\iota\,d(x, y, z))))$ is
uniform, while the code $\sigma_{\nattype}\,(\lambda x\,.\,\mathsf{if}\; x
= 17\;\mathsf{then}\; \delta_{B}(\lambda y\,.\,\iota\,
c(y))\;\mathsf{else}\; \sigma_{C(x)}\,(\lambda y\,.\,\delta_{D(x,
  y)}(\lambda z\,.\,\iota\,d(x, y, z))$ is not, since the shape is
sometimes $\sigma$-$\delta$-$\iota$ and sometimes
$\sigma$-$\sigma$-$\delta$-$\iota$. A precise description and study of
uniform $\pIR$ codes is out of scope of this paper; for further
information, we refer to our formal development~\cite{agda}.

The main construction that uniform $\pIR$ codes allow, while arbitrary
codes seem not to, is to construct a code for exponentiation of a
$\pIR$ code $\gamma$ with a set $K$, i.e. a code $K \to \gamma$ such
that $\inter{K \to \gamma}C \cong (K \to \inter{\gamma}_0C, f \mapsto
\Sigma k : K\,.\,\inter{\gamma}_1C(f(k)))$. Note how there is a sigma
type in the decoding of the family; as we have seen in
\Cref{ex:sigma-again}, families closed under $\Sigma$ are canonical
examples of a $\pIR$ construction. Since the construction of $K \to
\gamma$ depends on the definition of uniform codes, we do not give it
here, but refer again to our Agda formalisation~\cite{agda}, where we
also show that all constructions so far in this section result in
uniform codes (except for the coproduct of codes, whose construction
must be modified slightly). Given such a construction, we can now
interpret also container composition with nested functors as a $\pIR$
code by defining $\bullet : \Nest \to \pIR \to \pIR$ in the following
way:
\begin{align*}
  \Id \bullet \gamma &= \delta_1\,(\lambda X\,.\,X(\ast) \to \gamma) \\
  K\,(S, P) \bullet \gamma &= \sigma_S\,(\lambda s\,.\,P(s) \to \gamma) \\
  (N + N') \bullet \gamma &= (N \bullet \gamma) +_{\IR} (N' \bullet \gamma) \\
  (N \times N') \bullet \gamma &= (N \bullet \gamma) \times_{+} (N' \bullet \gamma) \\
  (N \circledast N') \bullet \gamma &= N  \bullet (N' \bullet \gamma)
\end{align*}

 \begin{lemma}
\label{thm:nestComp}
   $\inter{N \bullet \gamma} C = (\Int{N}_{\cont} C)
   \circ \inter{\gamma}C$. \qed
\end{lemma}

Putting everything together, we arrive at the main theorem of this
section.

\begin{theorem}
  For every $N : \Nest$, the initial algebra $\mu \Int{N}$ exists and
  is a container functor.
\end{theorem}
\begin{proof}
  By \cref{thm:nestIdK,thm:nestPlus,thm:nestTimes,thm:nestComp},
  $\Int{N}_{\cont}$ is an $\pIR$ functor. Hence by the results in
  \cref{sec:initial-alg}, it has an initial algebra which is a
  container $(S_N,P_N)$.  Since $\interC{-}$ preserves initial objects
  and filtered colimits of cartesian morphisms (Abbott~\cite{Abb},
  Propositions 4.5.1 and 4.6.7) and we know from
  \cref{thm:chain-splitting} in \cref{sec:initial-alg} that the
  initial algebra chain of an $\pIR$ functor is made from cartesian
  morphisms only, we can conclude that $\interC{(S_N,P_N)} = \mu
  \Int{N}$, showing that all nested types indeed are definable using
  containers.
\end{proof}

\section{Comparison to Plain \IR}\label{sec:oldIR}

We now investigate the relationship between $\pIR$ and $\IR$. On the
one hand we show in \cref{thm:IR-to-pos-IR} how to embed Dybjer and
Setzer's original coding scheme for $\IR$ into $\pIR$; this way we can
see $\IR$ as a subsystem of $\pIR$. On the other hand we show in
\cref{thm:pos-IR-to-IR} that on discrete categories, the two schemas
agrees having the same functorial interpretation; thus, using the
canonical embedding of the discretisation of a category into itself,
we can build a functor mapping $\pIR$ into $\IR$.

Note that every type $\D$ can be regarded as a discrete category,
which we by abuse of notation denote $|\D|$.  In the other direction,
every category $\C$ gives rise to a type $|\C|$ whose elements are the
objects of $\C$.
\begin{proposition}
\label{thm:IR-to-pos-IR}
There is a function $\varphi:\IR(\D) \to \pIR\,|\D|$ such that
\[
\inter{\varphi(\gamma)}_{\pIR\,|\D|} = \inter{\gamma}_{\IR(\D)}
\]
\end{proposition}
\begin{proof}
  The only interesting case is $\gamma = \delta_A\, F:\IR(\D)$; we
  define $\varphi(\delta_A\, F) = \delta_A\,(\varphi \circ F)$.  We
  need to ensure that $\varphi \circ F$ indeed is a functor, but since
  $|\D|$ is a discrete, so is $A \to |\D|$, and the mapping on objects
  $\varphi \circ F : (A \to |\D|) \to \pIR\,|\D|$ can trivially be
  extended to a functor $(A \to |\D|) \to \pIR\,|\D|$. It is easy to
  see that the two semantics do agree: on objects, the action is the
  same, and if $(h, \id)$ is a morphism in $\Fam\,|\D|$, we see from
  the definition of $\inter{\delta_A\, F}_{\pIR\,|\D|}(h, k)$ in the
  proof of \cref{thm:p-IR-functor} that
  \begin{align*}
       \inter{\delta_A\, (\varphi \circ F)}_{\pIR\,|\D|}(h, \id)
     &=
       [\inn_{h\circ g}\circ\inter{\varphi(F(Q\circ h\circ g)}(h, \id)
        \circ \inter{(\varphi \circ F)_{\to}(g^{*}(\id))}_{(X,P)}]_{g\,:A\to X} \\
        &=
        [\inn_{h\circ g}\circ\inter{\varphi(F(Q\circ h\circ g)}(h, \id) ]_{g\,:A\to X} \\
        &=
        [\inn_{h\circ g}\circ\inter{F(Q\circ h\circ g)}_{\IR(\D)}(h, \id) ]_{g\,:A\to X} \\
     &= \inter{\delta_A\, F}_{\IR(\D)}
  \end{align*}
  where $\inter{(\varphi \circ F)_{\to}(g^{*}(\id))}_{(X, P)} = \id$
  since $\inter{(\varphi \circ F)(-)}_{(X, P)}$ is a functor.
\end{proof}
This proposition shows that the theory of $\IR$ can be embedded in the
theory of $\pIR$. Some readers might perhaps be surprised that we only
define a function $\IR(\D) \to \pIR\,|\D|$, and not a functor. The
reason is the mismatch of morphisms between $\IR(\D)$ and
$\pIR\,|\D|$; because $\IR(\D)$ has a full and faithful embedding into
$\Fam\,|\D| \to \Fam\,|\D|$, whereas $\pIR\,|\D|$ has not, there are
necessarily morphisms in $\IR(\D)$ that have no counterpart in
$\pIR\,|\D|$. Going the other way, we are more successful, and can
make the previous result more precise: using the functoriality of
$\Fam$ (\Cref{rem:Fam}), we can embed $\Fam\,|\C|$ into $\Fam(\C)$. We
can then show that forgetting about the extra structure of $\C$ in
$\pIR$ simply gets us back to plain $\IR$.

\begin{proposition}
\label{thm:pos-IR-to-IR}
Let $\varepsilon : |\C|\to \C$ the canonical embedding of the
discretisation of a category $\C$ into itself. There is a functor
$\psi \,: \pIR\C \to \IR\,|\C|$ such that for all $\gamma : \pIR\C$
  \begin{equation}
    \label{eq:star}
\Fam(\varepsilon) \circ \inter{\psi(\gamma)}_{\IR\,|\C|} \cong \inter{\gamma}_{\pIR(\C)} \circ \Fam(\varepsilon) \tag{$\star$}
\end{equation}
Furthermore, $\psi \circ \varphi = \id$, where $\varphi:\IR(\C)\to
\pIR\,|\C|$ is the function from
\cref{thm:IR-to-pos-IR}. 
\end{proposition}
\begin{proof}
  We define the functor $\psi\,:\pIR(\C)\to \IR\,|\C|$ by recursion on
  the structure of $\gamma$. On objects, $\psi$ is defined as follows:
\begin{align*}
\psi(\iota\,c)&=\iota\,c\\
\psi(\sigma_{A} f)&=\sigma_A(a\mapsto \psi(f\,a))\\
\psi(\delta_A F)&= \delta_A(X\mapsto \psi(F(\varepsilon\circ X)))
\end{align*} 
We now use full and faithfulness of the interpretation functor
$\inter{\_}_{\IR\,|\C|}$, as proved in Ghani et~al.~\cite{Malatesta},
to let the function $\psi$ act on morphisms as well as on objects.
Since the two interpretation functors agree on objects, i.e.\
$\inter{\gamma}(X,P)=\inter{\psi(\gamma)}_{\IR\,|\C|}(X,P)$, a $\pIR$
morphism $\rho\,:\gamma\to \gamma'$ corresponds to a natural
transformation
$\inter{\rho}\,:\inter{\psi(\gamma)}_{\IR\,|\C|}\stackrel{\cdot}{\to}\inter{\psi(\gamma')}_{\IR\,|\C|}$. By
full and faithfulness of $\inter{\_}_{\IR\,|\C|}$ such a natural
transformation corresponds to an $\IR$ morphism $\psi(\gamma)\to
\psi(\gamma')$ which we take as the definition of
$\psi(\rho)$. Similarly, full and faithfulness of
$\inter{\_}_{\IR\,|\C|}$ ensure that composition and identity are
preserved by $\psi$, which is therefore a functor.
 
We are left with checking that \eqref{eq:star} holds for
morphisms. Recall from \cref{rem:Fam} that a morphism in $\Fam\,|\C|$
correspond to a split cartesian morphism in $\Fam(\C)$, i.e\ one whose
second component is an identity. Thus, to verify \eqref{eq:star} it is
enough to check that $\inter{\gamma}$ preserves such split cartesian
morphisms. The interesting case is $\gamma=\delta_A F$. Let
$(h,\id)\,:(X,P\circ h)\rightarrow (Y,P)$ be a morphism in
$\Fam\,|\C|$.  We have
\begin{align*}
\inter{\delta_A F}(h,\id)
&=[in_{h\circ g}\circ\inter{F(P\circ h\circ g)}(h,\id)\circ \inter{F_{\to}(g^{*}(\id))}_{(X,P)}]_{g\,:A\to X}\\
&=[in_{h\circ g}\circ\inter{F(P\circ h\circ g)}(h,\id) ]_{g\,:A\to X}
\end{align*}
where $\inter{F(g^*\id)}_{(X,P)} = \id$ since $g^*$, $F$ and
$\inter{\_}$ are functors.  By the induction hypothesis, each
$\inter{F(P\circ h\circ g)}(h,\id)$ is split cartesian. Furthermore
injections are split cartesian in $\Fam(\C)$, and since compositions
and cotuplings of split cartesian morphisms are still split
cartesian 
in $\Fam(\C)$ we conclude that $\inter{\delta_A\,F}(h,\id)$ indeed is
a split cartesian morphism as required.
 
Finally, since $\varepsilon$ is the identity on discrete categories
the two schemas agrees on discrete categories and we automatically get
$\psi \circ \varphi = \id$.
\end{proof}

\section{Existence of Initial Algebras}
\label{sec:initial-alg}

We briefly revisit the initial algebra argument used by Dybjer and
Setzer~\cite{DybSetz1}.  Inspecting their proof, we see that it indeed
is possible to adapt it also for the more general setting of positive
inductive-recursive definitions by making the appropriate adjustments.

Remember that a morphism $(h,k)\,:(U,T)\to (U,T')$ in $\Fam(\C)$ is a
\emph{split cartesian morphism} if $k = \id_T$, i.e.\ $T' \circ h =
T$, and that $\Fam\,|\C|$ is the subcategory (subfibration) of
$\Fam(\C)$ with the same objects, but with morphisms the split
cartesian ones only.
The proof of existence of initial algebras for $\IR$ functors as given
by Dybjer and Setzer~\cite{DybSetz1} takes place in the category
$\Fam\,|\C|$.  The hard work of the proof is divided between two
lemmas.  First Dybjer and Setzer prove that an $\IR$ functor
$\inter{\gamma}$ preserves $\kappa$-filtered colimits if $\kappa$ is
an inaccessible cardinal which suitably bounds the size of the index
sets in the image of the filtered diagram.  Secondly they use the
assumption of the existence of a large cardinal, namely a Mahlo
cardinal, to prove that such a cardinal bound for the index sets can
actually be found.  The exact definition of when a cardinal is a Mahlo
cardinal will not be important for the current presentation; see
Dybjer and Setzer~\cite{DybSetz1}, or the second author's
thesis~\cite{Lorenzothesis} for how this assumption is used.  The
existence of an initial algebra then follows a standard argument: the
initial algebra of a $\kappa$-continuous functor can be constructed as
the colimit of the initial chain up to $\kappa$ iterations (see e.g.\
Ad\'amek et al.~\cite{Adamek}).

Inspecting the proofs, we see that they crucially depend on morphisms
being split cartesian in several places.  Luckily, the morphisms
involved in the corresponding proofs for \pIR actually are!  As is
well-known, a weaker condition than $\kappa$-continuity is actually
sufficient: it is enough that the functor in question preserve the
specific colimit of the initial $\kappa$-chain.  We thus show that the
initial chain of a $\pIR$ functor actually lives in $\Fam\,|\C|$,
which will allow us to modify Dybjer and Setzer's proof accordingly.

\begin{lemma}
\label{thm:chain-splitting}
For each $\gamma:\pIR\,\C$, the initial chain  
\[
\zero \to \inter{\gamma}(\zero) \to \inter{\gamma}^{2}(\zero) \to \ldots
\]
consists of split cartesian morphisms only.
\end{lemma}
\begin{proof}
Recall that the connecting morphisms $\omega_{j,k}:\inter{\gamma}^{j}(\zero)\to \inter{\gamma}^{k}(\zero)$ are uniquely determined as follows:
\begin{itemize}
\item $\omega_{0,1} = {} !_{\inter{\gamma}(\zero)}$ is unique.
\item $\omega_{j+1,k+1}$ is $\inter{\gamma}(\omega_{j, k}) : \inter{\gamma}(\inter{\gamma}^{j}(\zero))\to \inter{\gamma}(\inter{\gamma}^{k}(\zero))$.
\item $\omega_{j,k}$ is the colimit cocone for $j$ a limit ordinal.
\end{itemize}
We prove the statement by induction on $j$.  It is certainly true that
$!_{\inter{\gamma}(\zero)} : (\zero, !) \to \inter{\gamma}(\zero)$ is
an identity at each component --- there are none. Thus $\omega_{0,1}$
is a split cartesian morphism.  At successor stages, we apply
\cref{thm:pos-IR-to-IR} and the induction hypothesis.  Finally, at
limit stages, we use the fact that the colimit lives in $\Fam\,|\C|$
and hence coincides with the colimit in that category on split
cartesian morphisms, so that the colimit cocone is split cartesian.
\end{proof}

Inspecting Dybjer and Setzer's original proof, we see that it now goes
through also for $\pIR$ if we insert appeals to
\cref{thm:chain-splitting} where necessary.  To finish the proof, we
also need to ensure that $\Fam(\C)$ has $\kappa$-filtered colimits;
this is automatically true if $\C$ has all small connected colimits
(compare \Cref{rem:Fam}), since $\Fam(\C)$ then is cocomplete.  Note
that discrete categories have all small connected colimits for trivial
reasons.

\begin{theorem}
  Assume that a Mahlo cardinal exists in the meta-theory. If $\C$ has
  connected colimits, then every functor $\inter{\gamma}$ for $\gamma
  : \pIR\,\C$ has an initial algebra. \qed
\end{theorem}

\section{Conclusions and Future Work}

In this paper we have introduced the theory $\pIR$ of positive
inductive-recursive definitions as a generalization of Dybjer and
Setzer's theory $\IR$ of inductive-recursive
definitions~\cite{DybSetz1,DybSetz2,DybSetz3}, different from the
fibrational generalization explored in Ghani et al.~\cite{GHMNS}: by
modifying both syntax and semantics of $\IR$ we have been able to
broaden the semantics to all of $\Fam(\C)$ and not just
$\Fam\,|\C|$. The theory of $\pIR$, with $\IR$ as a subtheory, paves
the way to the analysis of more sophisticated data types which allow
not only for the simultaneous definition of an inductive type $X$ and
of a recursive function $f:X\to \D$, but also takes the intrinsic
structure between objects in the target type $\D$ into account. This
is the case for example when $\D$ is a setoid, the category $\Set$ or
$\Dual{\Set}$, a groupoid or, even more generally, an arbitrary
category $\C$.

\sloppy In future work we aim to explore the theory of $\pIR$ from a
fibrational perspective: this will allow us to reconcile the theory of
$\pIR$ with the analysis of $\IR$ as given in Ghani~ et~
al.~\cite{GHMNS}. In particular this will amount to characterising the
semantics of $\delta$ codes as left Kan extensions. An open problem
for both $\pIR$ and $\IR$ is the question whether the definable
functors are closed under composition, i.e. if there is a code $\gamma
\circ \gamma'$ such that $\inter{\gamma \circ \gamma'} \cong
\inter{\gamma} \circ \inter{\gamma'}$ for all codes $\gamma$ and
$\gamma'$. Another interesting direction of research is to investigate
to which extent the rich structure of the families construction $\Fam$
will help shed light on the analysis of $\pIR$ types: in particular to
exploit the monadic structure of $\Fam$ and then to investigate the
relationship between the theory of $\pIR$ and the theory of familial
2-functors introduced by Weber~\cite{Weber}.
\label{sec:conclusion}

\bibliography{positive}
\bibliographystyle{alpha}

\end{document}

%% file: macros.tex
\newcommand{\ignore}[1]{}

\newcommand{\C}{\mathbb{C}}

\newcommand{\Fam}{\ensuremath{\mathsf{Fam}}}

\newcommand{\CAT}{\ensuremath{\mathsf{CAT}}}
\newcommand{\Set}{\ensuremath{\mathsf{Set}}}
\newcommand{\type}{\ensuremath{\mathsf{type}}}

\newcommand{\id}{\mathsf{id}}

\newcommand{\Dual}[1]{{#1}\ensuremath{^{\mathsf{op}}}}

\newcommand{\Id}{\mathsf{Id}}

\newcommand{\zero}{\mathsf{0}}
\newcommand{\one}{\mathsf{1}}
\newcommand{\two}{\mathsf{2}}

\newcommand{\D}{D}

\newcommand{\nattype}{\mathbb{N}}

\newcommand{\inl}{\mathsf{inl}}
\newcommand{\inr}{\mathsf{inr}}

\newcommand{\inn}{\mathsf{in}}

\newcommand{\IR}{\ensuremath{\mathsf{IR}}\xspace}
\newcommand{\pIR}{\ensuremath{\mathsf{IR}^{+}}\xspace}
\newcommand{\Hom}[1]{\ensuremath{\mathsf{Hom}_{#1}}\xspace}

\newcommand{\cont}{\ensuremath{\mathsf{Cont}}}

\newcommand{\inter}[1]{\llbracket #1 \rrbracket}
\newcommand{\interC}[1]{\llbracket #1 \rrbracket_{\mathsf{\cont}}}

\newcommand{\cons}[2]{\ensuremath{(\mathit{#1}\!\Rightarrow\!\mathit{#2})}}%

\newcommand{\Nest}{\mathsf{Nest}}
\newcommand{\Int}[1]{\llparenthesis#1\rrparenthesis}

\newcommand{\calF}{{\mathcal F}}

\newcommand{\NF}{\mathsf{NF}}
\newcommand{\pair}[2]{\langle#1,#2 \rangle}
\newcommand{\No}{\ensuremath{\mathsf{nf}}}
\newcommand{\Co}{\ensuremath{\mathsf{correct}}}

\theoremstyle{plain}

\newtheorem{theorem}[thm]{Theorem}

\newtheorem{lemma}[thm]{Lemma}
\newtheorem{proposition}[thm]{Proposition}

\theoremstyle{definition}

\newtheorem{remark}[thm]{Remark}
\newtheorem{remarks}[thm]{Remarks}
\crefname{remarks}{Remarks}{Remarks}
\newtheorem{example}[thm]{Example}
\newtheorem{definition}[thm]{Definition}
